\documentclass[preprint,12pt]{elsarticle}
\usepackage{amssymb}
\usepackage{tikz}
\usepackage[para,online,flushleft]{threeparttable}
\graphicspath{ {./figures/} }
\usepackage{hyperref}
\usepackage{float}
\usepackage{verbatim}
\restylefloat{figure}
\restylefloat{table}

\journal{ArXiv}

\usepackage{titlesec}
\titleclass{\subsubsubsection}{straight}[\subsection]
\newcounter{subsubsubsection}[subsubsection]
\renewcommand\thesubsubsubsection{\thesubsubsection.\arabic{subsubsubsection}}
\titleformat{\subsubsubsection}{\normalfont\normalsize\itshape}{\thesubsubsubsection.\space}{0em}{}
\titlespacing*{\subsubsubsection}{0pt}{2ex plus 1ex minus .2ex}{0.75ex plus .2ex}
\makeatletter
\def\toclevel@subsubsubsection{4}
\def\l@subsubsubsection{\@dottedtocline{4}{7em}{4em}}
\makeatother

\usepackage{amsmath}
\usepackage{amsthm,amssymb}
\usepackage{graphicx}
\usepackage{amsmath,amssymb,amsfonts}
\usepackage{algorithmic}
\usepackage{graphicx}
\usepackage{textcomp}
\usepackage{booktabs}
\usepackage{xcolor}
\usepackage{array}
\usepackage{multirow}
\usepackage{url}
\usepackage{float}
\usepackage{pifont}

\floatstyle{plaintop}
\restylefloat{table}
\usepackage{caption}
\usepackage{subcaption}
\newtheorem{Proposition}{Proposition}

\newtheorem{assumption}{Assumption}

\newtheorem{lemma}{Lemma}
\usepackage{booktabs}

\newif\ifblackandwhite
\blackandwhitetrue
\usepackage{pdflscape}
\usepackage{colortbl}

\usepackage{booktabs}

\usepackage{adjustbox}
\usepackage{rotating}

\def\BibTeX{{\rm B\kern-.05em{\sc i\kern-.025em b}\kern-.08em
    T\kern-.1667em\lower.7ex\hbox{E}\kern-.125emX}}

\usepackage{geometry}
\begin{document}
\begin{frontmatter}

\title{Opinion Dynamics in Social Multiplex Networks with Mono and Bi-directional Interactions in the Presence of Leaders}


\author{Amirreza Talebi$^{a}$, Sayed Pedram Haeri Boroujeni$^{b}$, Abolfazl Razi$^{c}$}

\affiliation{organization={Department of Integrated Systems Engineering},
            addressline={Ohio State University}, 
            city={Columbus},
            postcode={43210}, 
            state={OH},
            country={USA}\\
            Email: talebi.14@osu.edu}

\affiliation{organization={School of Computing},
            addressline={Clemson University}, 
            city={Clemson},
            postcode={29632}, 
            state={SC},
            country={USA}\\
            Email: shaerib@g.clemson.edu}

\affiliation{organization={School of Computing},
            addressline={Clemson University}, 
            city={Clemson},
            postcode={29632}, 
            state={SC},
            country={USA}\\
            Email: arazi@clemson.edu}

\begin{abstract}
We delve into the dynamics of opinions within a multiplex network using coordination games, where agents communicate either through one-way or two-way interactions, and where a designated leader may be present. By employing graph theory and Markov chains, we illustrate that despite zero diagonal elements in transition probability matrices or decomposable layers, opinions converge under specific conditions on the layers, leading to a consensus, and we show by a counterexample that a spanning tree in the union of the layers does not suffice. We further scrutinize the convergence rates of opinion dynamics in networks with one-way versus two-way interactions. We find that in networks with a designated leader, opinions converge towards the initial opinion of the leader, whereas in networks without a designated leader, opinions converge to a convex combination of the opinions of agents. Moreover, we emphasize the crucial role of designated leaders in steering opinion convergence within the network. Paired experiments on 1,200 random multiplex networks show that a leader slows down convergence, and that two-way interactions slow it down in the presence of a leader but speed it up in its absence.
\end{abstract}


\end{frontmatter}

\section{Introduction}

\subsection{Introduction}\label{intro:1}

Social networks play a crucial role in disseminating various types of information, including (mis)information, rumors, new technologies, and diseases \cite{kleinberg2007cascading}. Extensive research has focused on understanding opinion dynamics in single-layer networks, as evidenced by studies such as \cite{saber2003consensus,jadbabaie2003coordination,ghaderi2013opinion, ren2005consensus,sarlak2020approach,hendrickx2005convergence}.

In real-world scenarios, people engage across multiple communication platforms like Twitter, Facebook, Instagram, face-to-face interactions, among others. Each of these platforms can be viewed as a single-layer network or graph, where nodes represent individuals, and edges denote connections with varying weights reflecting the influence one places on the opinions of others. However, studying each isolated network fails to capture the true underlying network, comprised of multiple mono-layer networks and their inter-layer interactions, resulting in a loss of valuable information \cite{de2016physics}.

Therefore, for a comprehensive understanding of social networks, it is essential to consider their multi-layer nature \cite{de2016physics}. Multiplex networks, a subset of multi-layer networks, are characterized by nodes existing in multiple layers concurrently \cite{boccaletti2014structure}.

The dynamics of opinions within multi-layer networks can be viewed as coordination games. In this context, agents seek to optimize their payoffs, defined as minimizing the deviation from their opinions, by strategically selecting the best responses.
The exploration of coordination games on single-layer networks is well-documented in the literature, as seen in \cite{kleinberg2007cascading, ghaderi2013opinion,  fazeli2015duopoly}. Additionally, research has delved into coordination games on multi-layer networks, as evidenced by studies like \cite{gomez2012evolution, hu2017opinion, wang2015evolutionary}, with a comprehensive literature review available in \cite[p. 70]{boccaletti2014structure}.

These problems are also intricately linked to consensus problems where the agents need to agree on a common decision (opinion in this paper) with a rich body of literature such as \cite{jadbabaie2003coordination, blondel2005convergence, hendrickx2005convergence,  olshevsky2006convergence, olshevsky2009convergence}.

Furthermore, certain studies have delved into the notion of on and off layers within multi-layer networks, where specific layers are active during certain time steps while others remain inactive, as demonstrated in \cite{ding2017asynchronous, dong2017dynamics}.

In our paper, we formulated a model for coordination games on a two-layer multiplex network, incorporating the on and off layers concepts based on models from \cite{ghaderi2013opinion, ding2017asynchronous, dong2017dynamics} designed for single-layer networks.  We further examined the convergence and convergence rate of the model, considering the mono or bidirectional interactions between agents. We show by a counterexample that a spanning tree in the union of both layers does not guarantee consensus when agents may place no weight on their own opinions, and we give conditions on the layers that do. This problem exhibits a myriad of applications across diverse fields, including medical networks (e.g., \cite{hamzehconceptual, hamzehnew}), and neural networks (e.g., \cite{https://doi.org/10.13140/rg.2.2.29416.03849,soleymani_forecasting_2024,boroujeni2024ic, 9814765}). 

The remainder of the paper is organized as follows: a comprehensive literature review is presented first. Then, Section \ref{sec:model} introduces the model and preliminaries, followed by Section \ref{sec:Convergence}, which offers an analytical analysis of the convergence of the model to the equilibrium. Section \ref{sec:Convergence Rate} explores the convergence rate of the model, section \ref{sec: numerical} explains the numerical results, and the concluding section summarizes the findings.

\subsection{Related Work}\label{sec:rel}

In the domain of game-theoretic models, networked-coordination games \cite{kleinberg2007cascading} play a crucial role in elucidating how individuals in online settings become informed about new technologies, opinions, and rumors. While addressing the seeding problem, crucial in viral marketing \cite{fazeli2015duopoly},  \cite{kleinberg2007cascading} did not explore the convergence rate of information evolution. This aspect was later addressed by \cite{montanari2009convergence} in the context of a single-layer network, noting slow convergence in well-connected networks and faster convergence in poorly connected networks. In contrast, \cite{ghaderi2013opinion} presented different results, revealing faster convergence in complete graphs compared to poorly connected graphs like rings.

In the multi-layer context, \cite{gomez2012evolution} delved into the dynamics of an evolutionary game on multiplex networks, highlighting enhanced resilience of cooperative behaviors compared to single-layer settings. Their investigation into the prisoner's dilemma on multiplex networks indicated that the reward for defecting influences cooperative behaviors differently than in single-layer settings. The reference \cite{hu2017opinion} explored the influence of multi-layer networks on opinion diffusion with and without stubborn agents, placing emphasis on strongly connected and closed agent sets leading to opinion convergence. Stubborn agents are those who never or hardly change their opinions.  

Consensus problems, dealing with group agreement, have undergone extensive scrutiny \cite{ren2005consensus, olfati2004consensus}. The reference \cite{olfati2004consensus} delved into dynamic agent consensus under various assumptions for directed and undirected graphs. The reference \cite{ren2005consensus} demonstrated consensus in a single directed network under network topology switching, provided the union of network topologies forms a spanning tree. References \cite{blondel2005convergence} and \cite{hendrickx2005convergence} scrutinized convergence rates of bidirectional equal neighbor models, in which agents in the underlying communication graph put equal weights on their neighboring agents' opinions, and infinite products of type symmetric stochastic matrices with positive diagonals. 

In the context of consensus problems, the study often involves products of stochastic matrices, tracing back to the works of \cite{wolfowitz1963products} and  \cite{hajnal1958weak}. The reference \cite{xia2015products} underscored that achieving consensus necessitates the left-sided product sequence of stochastic matrices to converge to a rank one matrix. Recent work by \cite{xia2018generalized} introduced type-1 and type-2 generalized Sarymsakov matrices, sets of stochastic matrices whose product of compact subsets results in a rank one matrix. Our unique contribution lies in demonstrating the convergence rate of these matrix multiplications under specific circumstances. Our methodology is rooted in a coordination game-based model within the realm of multi-layered networked systems and Markov chains.

In this paper, we are constructing a two-layer multiplex network where agents engage in interactions solely on one layer during odd time steps and on both layers during even time steps. Our findings indicate that, under Assumption~\ref{as:strongly-connected}, consensus will eventually be achieved in the network. It is crucial to note the distinction between the convergence of opinions and consensus: while in consensus, all agents ultimately adopt the same opinion, convergence allows for the possibility that some agents may maintain different opinions.

In our model, we introduce a leader-follower structure, where an agent without influencers but influencing others is deemed a leader, and a follower is an agent influenced by the leader. Furthermore, we explore the convergence rates of networks with mono-directional or bidirectional interactions, both with and without leaders. To our knowledge, the model incorporating these specific elements and assumptions is unprecedented, and there is no well-established body of work studying such settings.

\section{Model and Preliminaries}\label{sec:model}

\subsection{The multiplex network}\label{sec:network-model}

We examine the convergence of opinions within a two-layer multiplex network consisting of $n$ agents present on both layers. The multiplex network is denoted as $\mathcal{M}=\{\mathcal{V},\mathcal{G}\}$, where $\mathcal{V}=\{v_1, \ldots, v_n\}$ represents the set of agents, and $v_i\in \mathcal{V}, i\in \mathcal{I}=\{1, \ldots, n\}$ denotes agent $i$. The network's layers are represented by the set of weighted, directed graphs $\mathcal{G}=\{\mathcal{G}_{\alpha}=(\mathcal{V},\mathcal{E}_{\alpha},\mathcal{W}_{\alpha})\,|\, \mathcal{E}_{\alpha}\subseteq \mathcal{V} \times \mathcal{V}, \mathcal{W}_{\alpha}: \mathcal{E}_{\alpha}\rightarrow \mathbb{R}_{\geq 0}, \alpha \in \{1, 2\}\}$.

Additionally, the set of neighbors of agent $i$ in layer $\alpha$ is denoted as $\delta_{i\alpha} = \{v_j| (v_j,v_i)\in \mathcal{E}_{\alpha}, j\in \mathcal{I}\}$, indicating that an edge $(v_j, v_i)$ implies agent $i$'s opinion is influenced by agent $j$'s opinion; when no confusion arises, we identify $v_j$ with its index $j$. Thus, the set of neighbors for an agent encompasses all agents influencing the opinion of that specific agent. Moreover, in our model, the interactions between agents can be mono or bi-directional. In the case of mono-directional interactions, an agent's opinion is influenced by its neighbors, but it may not necessarily affect the opinions of its neighbors. If agent $i$ has no neighbor in layer 2, we set $\delta_{i2}:=\delta_{i1}$, i.e., the agent keeps listening to layer 1 when both layers are active.

On a layer $\alpha$, a leader is identified as the agent $i$ for whom $\delta_{i\alpha} = \{v_i\}$, meaning the leader has no neighbors other than itself. A leader on both layers is a leader in the union of both layers.
It's noteworthy that a leader in one layer may act as a follower in another layer. For example in Figure \ref{fig:111}, agent $A$ is the leader on the first layer, perhaps serving as the administrator of a channel on a social platform, while concurrently acting as a follower on the second layer.

\begin{figure}[H]
\centering
\tikzset{every picture/.style={line width=0.75pt}} 

\begin{tikzpicture}[x=0.75pt,y=0.75pt,yscale=-0.75,xscale=0.75]

\draw   (135,39) -- (381,39) -- (311,135) -- (65,135) -- cycle ;
\draw   (109,109.5) .. controls (109,102.6) and (114.6,97) .. (121.5,97) .. controls (128.4,97) and (134,102.6) .. (134,109.5) .. controls (134,116.4) and (128.4,122) .. (121.5,122) .. controls (114.6,122) and (109,116.4) .. (109,109.5) -- cycle ;
\draw   (234,59.5) .. controls (234,52.6) and (239.6,47) .. (246.5,47) .. controls (253.4,47) and (259,52.6) .. (259,59.5) .. controls (259,66.4) and (253.4,72) .. (246.5,72) .. controls (239.6,72) and (234,66.4) .. (234,59.5) -- cycle ;
\draw   (318,58.5) .. controls (318,51.6) and (323.6,46) .. (330.5,46) .. controls (337.4,46) and (343,51.6) .. (343,58.5) .. controls (343,65.4) and (337.4,71) .. (330.5,71) .. controls (323.6,71) and (318,65.4) .. (318,58.5) -- cycle ;
\draw   (165,66.5) .. controls (165,59.6) and (170.6,54) .. (177.5,54) .. controls (184.4,54) and (190,59.6) .. (190,66.5) .. controls (190,73.4) and (184.4,79) .. (177.5,79) .. controls (170.6,79) and (165,73.4) .. (165,66.5) -- cycle ;
\draw   (219,111.5) .. controls (219,104.6) and (224.6,99) .. (231.5,99) .. controls (238.4,99) and (244,104.6) .. (244,111.5) .. controls (244,118.4) and (238.4,124) .. (231.5,124) .. controls (224.6,124) and (219,118.4) .. (219,111.5) -- cycle ;
\draw    (134,108.5) -- (167.58,75.4) ;
\draw [shift={(169,74)}, rotate = 135.41] [color={rgb, 255:red, 0; green, 0; blue, 0 }  ][line width=0.75]    (10.93,-3.29) .. controls (6.95,-1.4) and (3.31,-0.3) .. (0,0) .. controls (3.31,0.3) and (6.95,1.4) .. (10.93,3.29)   ;
\draw    (188,75) -- (220.5,103.68) ;
\draw [shift={(222,105)}, rotate = 221.42] [color={rgb, 255:red, 0; green, 0; blue, 0 }  ][line width=0.75]    (10.93,-3.29) .. controls (6.95,-1.4) and (3.31,-0.3) .. (0,0) .. controls (3.31,0.3) and (6.95,1.4) .. (10.93,3.29)   ;
\draw    (259,59.5) .. controls (286.58,53.1) and (273.41,104.42) .. (247.68,73.47) ;
\draw [shift={(246.5,72)}, rotate = 52.07] [color={rgb, 255:red, 0; green, 0; blue, 0 }  ][line width=0.75]    (10.93,-3.29) .. controls (6.95,-1.4) and (3.31,-0.3) .. (0,0) .. controls (3.31,0.3) and (6.95,1.4) .. (10.93,3.29)   ;
\draw    (244,111.5) -- (319.27,68) ;
\draw [shift={(321,67)}, rotate = 149.98] [color={rgb, 255:red, 0; green, 0; blue, 0 }  ][line width=0.75]    (10.93,-3.29) .. controls (6.95,-1.4) and (3.31,-0.3) .. (0,0) .. controls (3.31,0.3) and (6.95,1.4) .. (10.93,3.29)   ;
\draw   (131,143) -- (377,143) -- (307,239) -- (61,239) -- cycle ;
\draw   (105,213.5) .. controls (105,206.6) and (110.6,201) .. (117.5,201) .. controls (124.4,201) and (130,206.6) .. (130,213.5) .. controls (130,220.4) and (124.4,226) .. (117.5,226) .. controls (110.6,226) and (105,220.4) .. (105,213.5) -- cycle ;
\draw   (230,163.5) .. controls (230,156.6) and (235.6,151) .. (242.5,151) .. controls (249.4,151) and (255,156.6) .. (255,163.5) .. controls (255,170.4) and (249.4,176) .. (242.5,176) .. controls (235.6,176) and (230,170.4) .. (230,163.5) -- cycle ;
\draw   (314,162.5) .. controls (314,155.6) and (319.6,150) .. (326.5,150) .. controls (333.4,150) and (339,155.6) .. (339,162.5) .. controls (339,169.4) and (333.4,175) .. (326.5,175) .. controls (319.6,175) and (314,169.4) .. (314,162.5) -- cycle ;
\draw   (161,170.5) .. controls (161,163.6) and (166.6,158) .. (173.5,158) .. controls (180.4,158) and (186,163.6) .. (186,170.5) .. controls (186,177.4) and (180.4,183) .. (173.5,183) .. controls (166.6,183) and (161,177.4) .. (161,170.5) -- cycle ;
\draw   (215,215.5) .. controls (215,208.6) and (220.6,203) .. (227.5,203) .. controls (234.4,203) and (240,208.6) .. (240,215.5) .. controls (240,222.4) and (234.4,228) .. (227.5,228) .. controls (220.6,228) and (215,222.4) .. (215,215.5) -- cycle ;
\draw    (130,212.5) -- (163.58,179.4) ;
\draw [shift={(165,178)}, rotate = 135.41] [color={rgb, 255:red, 0; green, 0; blue, 0 }  ][line width=0.75]    (10.93,-3.29) .. controls (6.95,-1.4) and (3.31,-0.3) .. (0,0) .. controls (3.31,0.3) and (6.95,1.4) .. (10.93,3.29)   ;
\draw    (184,179) -- (216.5,207.68) ;
\draw [shift={(218,209)}, rotate = 221.42] [color={rgb, 255:red, 0; green, 0; blue, 0 }  ][line width=0.75]    (10.93,-3.29) .. controls (6.95,-1.4) and (3.31,-0.3) .. (0,0) .. controls (3.31,0.3) and (6.95,1.4) .. (10.93,3.29)   ;
\draw    (240,215.5) -- (315.27,172) ;
\draw [shift={(317,171)}, rotate = 149.98] [color={rgb, 255:red, 0; green, 0; blue, 0 }  ][line width=0.75]    (10.93,-3.29) .. controls (6.95,-1.4) and (3.31,-0.3) .. (0,0) .. controls (3.31,0.3) and (6.95,1.4) .. (10.93,3.29)   ;
\draw    (314,162.5) -- (257,163.47) ;
\draw [shift={(255,163.5)}, rotate = 359.03] [color={rgb, 255:red, 0; green, 0; blue, 0 }  ][line width=0.75]    (10.93,-3.29) .. controls (6.95,-1.4) and (3.31,-0.3) .. (0,0) .. controls (3.31,0.3) and (6.95,1.4) .. (10.93,3.29)   ;
\draw    (134,108.5) .. controls (161.58,102.1) and (148.41,153.42) .. (122.68,122.47) ;
\draw [shift={(121.5,121)}, rotate = 52.07] [color={rgb, 255:red, 0; green, 0; blue, 0 }  ][line width=0.75]    (10.93,-3.29) .. controls (6.95,-1.4) and (3.31,-0.3) .. (0,0) .. controls (3.31,0.3) and (6.95,1.4) .. (10.93,3.29)   ;
\draw    (230,163.5) -- (131.79,212.61) ;
\draw [shift={(130,213.5)}, rotate = 333.43] [color={rgb, 255:red, 0; green, 0; blue, 0 }  ][line width=0.75]    (10.93,-3.29) .. controls (6.95,-1.4) and (3.31,-0.3) .. (0,0) .. controls (3.31,0.3) and (6.95,1.4) .. (10.93,3.29)   ;
\draw    (318,58.5) -- (192,66.38) ;
\draw [shift={(190,66.5)}, rotate = 356.42] [color={rgb, 255:red, 0; green, 0; blue, 0 }  ][line width=0.75]    (10.93,-3.29) .. controls (6.95,-1.4) and (3.31,-0.3) .. (0,0) .. controls (3.31,0.3) and (6.95,1.4) .. (10.93,3.29)   ;
\draw    (314,162.5) -- (188,170.38) ;
\draw [shift={(186,170.5)}, rotate = 356.42] [color={rgb, 255:red, 0; green, 0; blue, 0 }  ][line width=0.75]    (10.93,-3.29) .. controls (6.95,-1.4) and (3.31,-0.3) .. (0,0) .. controls (3.31,0.3) and (6.95,1.4) .. (10.93,3.29)   ;

\draw (115,100) node [anchor=north west][inner sep=0.75pt]   [align=left] {\scriptsize{A}};
\draw (172,57) node [anchor=north west][inner sep=0.75pt]   [align=left] {\scriptsize{B}};
\draw (241,50) node [anchor=north west][inner sep=0.75pt]   [align=left] {\scriptsize{C}};
\draw (324,49) node [anchor=north west][inner sep=0.75pt]   [align=left] {\scriptsize{D}};
\draw (226,103) node [anchor=north west][inner sep=0.75pt]   [align=left] {\scriptsize{E}};
\draw (111,204) node [anchor=north west][inner sep=0.75pt]   [align=left] {\scriptsize{A}};
\draw (168,161) node [anchor=north west][inner sep=0.75pt]   [align=left] {\scriptsize{B}};
\draw (237,154) node [anchor=north west][inner sep=0.75pt]   [align=left] {\scriptsize{C}};
\draw (320,153) node [anchor=north west][inner sep=0.75pt]   [align=left] {\scriptsize{D}};
\draw (222,207) node [anchor=north west][inner sep=0.75pt]   [align=left] {\scriptsize{E}};
\draw (115,66) node [anchor=north west][inner sep=0.75pt]   [align=left] {{\tiny Layer 1}};
\draw (115,165) node [anchor=north west][inner sep=0.75pt]   [align=left] {{\tiny Layer 2}};

\end{tikzpicture}
\caption{An example of a two-layer multiplex network: Agent $A$ assumes a leadership role on Layer 1 but not on Layer 2. The neighbors of agent $B$ include agents $A$ and $D$ on both layers.
\label{fig:111}
}
\end{figure}

\subsection{The adjacency matrices} 

In this paper, agents participate in interactions over an infinite, discrete time horizon indexed by $t$. We assume that layer 1 is consistently active, implying continuous influence among agents at every time step. In contrast, layer 2 becomes active exclusively during even time steps, precisely at instances where $t\in \{2, 4, 6, \ldots\}$. This modeling decision allows for the depiction of scenarios where an agent frequently communicates with friends, while interactions with their superior occur infrequently. This situation is analogous to cases where network topology experiences temporal changes.

We define the weighted adjacency matrix $\mathcal{A}_\alpha$ of a single layer $\alpha$ by
\begin{align}\label{eqn:ad1}
\mathcal{A}_{\alpha}= \begin{cases}
[\mathcal{A}_{\alpha}]_{ij}= \frac{1}{|\delta_{i\alpha}|},& \text{$\forall j \in \delta_{i\alpha}, \forall i \in \mathcal{I},$}\\
[\mathcal{A}_{\alpha}]_{ij}=0, & \text{$\forall j \notin \delta_{i\alpha}, \forall i \in \mathcal{I}$}
\end{cases}
\end{align}
where $[\mathcal{A}_{\alpha}]_{ij}$ is the $ij^{\textsuperscript{th}}$ element of matrix $\mathcal{A}_{\alpha}$. To clarify, all the neighbors of an agent exert equal influence on the agent's opinion. Additionally, an agent with only a self-loop in the network indicates that the agent is both its own neighbor and a leader, influencing only its own opinion.

 The adjacency matrix of the Multiplex network will be defined as follows: 
  \begin{equation}
    A(t)=
    \begin{cases}
      \frac{1}{2}(\mathcal{A}_1+\mathcal{A}_2), &  \text{if}\ t\equiv 0\ \mod\ 2\\
      \mathcal{A}_1, & \text{Otherwise.}
    \end{cases}\label{eqn:m2}
  \end{equation} 
  
In words, the agent puts equal weight on the importance of both layers when they are both active. 
Moreover, under Assumption~\ref{as:strongly-connected}.1 below, matrices $\mathcal{A}_1$, $\mathcal{A}_2$ and $A(t)$, are row stochastic matrices i.e., $\sum_{j}[\mathcal{A}_1]_{ij} = \sum_{j}[\mathcal{A}_2]_{ij} = \sum_{j}[A(t)]_{ij} =1, \ \forall i,t$. 

\subsection{Agents' opinion dynamics}\label{sec:opinion-model}

We represent the opinion of agent $i$ at time step $t$ as $x_i(t) \in [0,10]$, and denote the vector of opinions of all agents at time step $t$ as $\mathbf{x}(t) = [x_1(t), \ldots, x_n(t)]^\top$. The initial opinions profile of the agents is denoted by $\mathbf{x}(0)$.

We extend the coordination game proposed in \cite{ghaderi2013opinion} for single-layer networks to a multiplex network. In this game, at time $t+1$, agent $i$ chooses its new opinion $y$ to minimize the following cost function, given the opinions $\mathbf{x}(t)$:
\begin{equation}\label{eqn:cost}
    J_i(y;\mathbf{x}(t))=
    \begin{cases}
    \frac{1}{4}\sum_{\alpha=1}^{2}\frac{1}{|\delta_{i\alpha}|}\sum_{j\in\delta_{i\alpha}}(y-x_j(t))^2,\text{ if }
    t+1\equiv 0 \mod\ 2\\
    \\
    \frac{1}{2}\sum_{j\in \delta_{i1}}(y-x_j(t))^2,\ \text{Otherwise}. \\
    \end{cases}
\end{equation}

In simple terms, agent $i$ seeks to minimize the difference between its opinion and that of its neighbors. The provided cost function reflects how this set of neighbors changes based on the active layers. The factors $1/|\delta_{i\alpha}|$ let the agent weigh both layers equally, as in \eqref{eqn:m2}, and a self-loop $v_i\in\delta_{i\alpha}$ penalizes deviations from the agent's current opinion.

Since $J_i$ is strictly convex in $y$, the first-order condition $\frac{d J_i(y;\mathbf{x}(t))}{d y}=0$ gives the best response $x_i(t+1)=\sum_{j}[A(t+1)]_{ij}x_j(t)$ of agent $i$ at time $t+1$, i.e.,
\begin{align}\label{eqn:opinion-dynamics}
    \mathbf{x}(t+1) = A(t+1)\mathbf{x}(t),
\end{align}
so that $\mathbf{x}(t)=A(t)\cdots A(1)\mathbf{x}(0)\in[0,10]^n$.
 Note that this model is very well-known in the literature (e.g., see \cite{jadbabaie2003coordination, ren2005consensus, hendrickx2005convergence, olshevsky2009convergence}). Mostly it is called the linearized version of the Vicsek model, \cite{Vicsek}. However, we study this model under special settings. 
We make the following assumptions for our Multiplex network.
\begin{assumption}\label{as:strongly-connected}
\begin{enumerate}
    \item Every agent without neighbors in layer 1 possesses a self-loop in layer 1, i.e., $\delta_{i1}=\{v_i\}$.
    \item Every closed communication class of layer 1 is aperiodic.
    \item The graph $\tilde{\mathcal{G}}$, obtained from layer 1 by adding a link $(v_f,v_i)$ whenever $v_f$ lies in an aperiodic communication class of layer 1 and layer 1 contains a directed path (possibly of length zero) from $v_f$ to a neighbor $v_j\neq v_i$ of agent $i$ in layer 2, contains a directed spanning tree.
\end{enumerate}
\end{assumption}
In graph theory, a directed spanning tree rooted in a vertex of a directed graph is a directed tree where all the vertices are reachable from the root vertex (see \cite{bapat2010graphs}). Furthermore, we define a communication class of a layer as a maximal set of agents in which a directed path exists from each agent to any other within the set; an agent that lies on no cycle forms a class of its own. A class is \emph{closed} if its agents have no neighbor outside the class in that layer, and \emph{aperiodic} if the greatest common divisor of the lengths of its directed cycles is one (a self-loop is a cycle of length one), e.g., if it contains an agent with a self-loop; a class without cycles is not aperiodic.

Every link $(v_f,v_i)$ of $\tilde{\mathcal{G}}$ corresponds to a directed path in the union of both layers. Hence, Assumption~\ref{as:strongly-connected}.3 implies that the union of both layers contains a directed spanning tree; in particular, there is at most one leader in the union of both layers, and such a leader is the root. If every communication class of layer 1 contains an agent with a self-loop, then every link of layer 2 other than a self-loop belongs to $\tilde{\mathcal{G}}$, and Assumption~\ref{as:strongly-connected}.3 reduces to a directed spanning tree in the union of both layers. Section~\ref{sec:Convergence} shows that the latter does not suffice in general.

It is worth noting that \cite{ren2005consensus} assumed that the union of the graphs contains a spanning tree and that all nodes in a graph must have a self-loop. In our work, the diagonal entries of $\mathcal{A}_1$ and $A(t)$ may be zero, and layer 1 may be decomposable, at the price of Assumptions~\ref{as:strongly-connected}.2 and \ref{as:strongly-connected}.3. The absence of a self-loop for a node allows for the possibility of fully open-minded or fully submissive agents, meaning the agent places no weight on its opinion. Additionally, they established convergence by showing that the products of the adjacency matrices are stochastic, indecomposable, and aperiodic (SIA) matrices, which relies on the self-loops; we prove that $A(2)\mathcal{A}_1$ is SIA without them. Furthermore, the (row) stochastic matrix $A$ is defined by \cite{wolfowitz1963products} to be \emph{SIA} if there exists a rank one matrix $L$ such that:
    \[L = \lim_{\ell\rightarrow \infty} A^\ell~.\]

\section{Convergence of the Opinion Profile}\label{sec:Convergence}

In this section, we demonstrate that the opinion dynamics described by \eqref{eqn:opinion-dynamics} will lead to convergence and consensus. Since layer 2 is active at even time steps, $\mathbf{x}(2\ell)=C^\ell\mathbf{x}(0)$ and $\mathbf{x}(2\ell+1)=\mathcal{A}_1C^\ell\mathbf{x}(0)$ for $\ell\geq 0$, where $C:=A(2)\mathcal{A}_1$. We establish the convergence for two scenarios: one where there is one leader in the union of both layers, and another where there are no leaders in the union of both layers. It's important to note that each layer can have multiple leaders, but Assumption~\ref{as:strongly-connected} allows at most one leader in the union of both layers. Several leaders in the union of both layers keep their initial opinions; consequently, consensus will not be reached in such a scenario unless these opinions coincide.

A directed spanning tree in the union of both layers does not suffice when the diagonal entries of $\mathcal{A}_1$ and $A(2)$ may be zero. For example, let three agents $p$, $q$, and $r$ have the neighbors $\delta_{p1}=\delta_{q1}=\{v_p\}$, $\delta_{r1}=\delta_{r2}=\{v_r\}$, $\delta_{p2}=\{v_q\}$, and $\delta_{q2}=\{v_r\}$. Then $r$ is the only leader and the root of a directed spanning tree of the union of both layers. However, $x_q(2\ell+1)=x_p(2\ell)$ and, hence, $x_p(2\ell+2)=\frac{1}{2}(x_p(2\ell+1)+x_q(2\ell+1))=x_p(2\ell)$, so that $x_p(2\ell)=x_p(0)$ and $x_r(2\ell)=x_r(0)$ for all $\ell$, and consensus is not reached if $x_p(0)\neq x_r(0)$. This network violates Assumption~\ref{as:strongly-connected}.3: agent $q$ lies on no cycle of layer 1, and $\tilde{\mathcal{G}}$ only contains the links $(v_p,v_p)$, $(v_p,v_q)$, $(v_r,v_q)$, and $(v_r,v_r)$.

For a non-negative matrix $P$, we call the graph with a link $(v_j,v_i)$ whenever $[P]_{ij}>0$ the graph of $P$, so that the graph of $\mathcal{A}_\alpha$ is layer $\alpha$, and we use the notions of communication classes of Section~\ref{sec:model} for it. We write $[P]_i$ for row $i$ of $P$ and $\mathbf{e}_i$ for the $i$th unit vector.

\begin{lemma}\label{lem:sia}
Under Assumption~\ref{as:strongly-connected}, the matrix $C=A(2)\mathcal{A}_1$ is SIA. More precisely, the graph of $C$ has exactly one closed communication class $\mathcal{Z}$, which is aperiodic, and $\lim_{\ell\rightarrow\infty}C^\ell=\mathbf{1}\pi^\top$, where $\pi$ is the stationary distribution of $C$ and $[\pi]_i>0$ if and only if $v_i\in\mathcal{Z}$.
\end{lemma}

\begin{proof}
Products of row stochastic matrices are row stochastic \cite{horn2012matrix}; hence, $C$ is stochastic. Let $\Lambda$ be the directed graph with node set $\mathcal{V}\times\{0,1\}$, a link from $(v_i,0)$ to $(v_j,1)$ if $[A(2)]_{ij}>0$, and a link from $(v_j,1)$ to $(v_w,0)$ if $[\mathcal{A}_1]_{jw}>0$. As $C$ is a product of non-negative matrices, $[C^\ell]_{iw}>0$ if and only if $\Lambda$ contains a walk of length $2\ell$ from $(v_i,0)$ to $(v_w,0)$; walks in $\Lambda$ follow the links of the multiplex network backwards, from an agent to its neighbors. Fix an agent $v_i$, and call an agent $v_w$ \emph{reached} if $\Lambda$ contains walks from $(v_i,0)$ to both $(v_w,0)$ and $(v_w,1)$.

(a) Since $[A(2)]_{uv}\geq\frac{1}{2}[\mathcal{A}_1]_{uv}$, every $v_v\in\delta_{u1}$ yields links from $(v_u,0)$ to $(v_v,1)$ and from $(v_u,1)$ to $(v_v,0)$. Hence, if $v_w$ is reached and $v_v\in\delta_{w1}$, then $v_v$ is reached. Moreover, if $\Lambda$ contains a walk from $(v_i,0)$ to $(v_w,\phi)$ for some $\phi\in\{0,1\}$ and $v_w$ lies in an aperiodic communication class $\mathcal{K}$ of layer 1, then every agent of $\mathcal{K}$ is reached, since $\mathcal{K}$ contains a cycle of odd length and, hence, layer 1 contains walks of both parities within $\mathcal{K}$ between any two of its agents.

(b) Let $v_w$ be reached, let $v_j\in\delta_{w2}$ with $j\neq w$, and let $v_f$ lie in an aperiodic communication class of layer 1, and let layer 1 contain a directed path from $v_f$ to $v_j$. Since $[A(2)]_{wj}\geq\frac{1}{2}[\mathcal{A}_2]_{wj}>0$, $\Lambda$ contains a link from $(v_w,0)$ to $(v_j,1)$, and following the path from $v_j$ back to $v_f$ by (a) gives a walk from $(v_i,0)$ to $(v_f,\phi)$ for some $\phi$. By (a), $v_f$ is reached.

(c) Every agent has a neighbor in layer 1 (Assumption~\ref{as:strongly-connected}.1); hence, following neighbors in layer 1 from $v_i$ leads into a closed communication class $\mathcal{K}$ of layer 1, which is aperiodic by Assumption~\ref{as:strongly-connected}.2. By (a), all agents of $\mathcal{K}$ are reached.

By (a) and (b), if $v_w$ is reached and $(v_f,v_w)$ is a link of $\tilde{\mathcal{G}}$, then $v_f$ is reached. Let $v_r$ be the root of a directed spanning tree of $\tilde{\mathcal{G}}$ (Assumption~\ref{as:strongly-connected}.3). Following the path of $\tilde{\mathcal{G}}$ from $v_r$ to an agent of $\mathcal{K}$ backwards shows that $v_r$ is reached. As $v_i$ was arbitrary, the graph of $C$ contains a directed path from $v_r$ to every agent. Hence, every closed communication class of $C$ contains $v_r$, and $C$ has exactly one closed communication class $\mathcal{Z}$. Applying (c) to $v_i=v_r$, all agents of a closed communication class $\mathcal{K}'$ of layer 1 are reached from $v_r$, so that $\mathcal{K}'\subseteq\mathcal{Z}$. As $\mathcal{K}'$ is closed, aperiodic, and strongly connected in layer 1, the principal submatrix of $\mathcal{A}_1$ indexed by $\mathcal{K}'$ is stochastic and primitive \cite[Section 8.5]{horn2012matrix}; hence, $[\mathcal{A}_1^m]_{uu}>0$ for $v_u\in\mathcal{K}'$ and all large $m$. Since $C\geq\frac{1}{2}\mathcal{A}_1^2$ entrywise, $[C^\ell]_{uu}\geq 2^{-\ell}[\mathcal{A}_1^{2\ell}]_{uu}>0$ for all large $\ell$; thus, $\mathcal{Z}$ contains closed walks of lengths $\ell$ and $\ell+1$, i.e., it is aperiodic. Finally, a stochastic matrix with exactly one closed communication class, which is aperiodic, satisfies $\lim_{\ell\rightarrow\infty}C^\ell=\mathbf{1}\pi^\top$ with $[\pi]_i>0$ if and only if $v_i\in\mathcal{Z}$ \cite[Chapter 4]{seneta2006non}; hence, $C$ is SIA.
\end{proof}

\begin{Proposition}\label{prop:convergence}
Under Assumption~\ref{as:strongly-connected}, for every $\mathbf{x}(0)$, $\lim_{t\rightarrow\infty}\mathbf{x}(t)=\bar{\mathbf{x}}:=\mathbf{1}\pi^\top\mathbf{x}(0)$, with $\pi$ as in Lemma~\ref{lem:sia}. If $v_r$ is a leader in the union of both layers, then $\pi=\mathbf{e}_r$, i.e., all opinions converge to $x_r(0)$; otherwise, $\bar{\mathbf{x}}$ is a convex combination of the initial opinions of the agents in $\mathcal{Z}$.
\end{Proposition}

\begin{proof}
By Lemma~\ref{lem:sia}, $\mathbf{x}(2\ell)=C^\ell\mathbf{x}(0)\rightarrow\mathbf{1}\pi^\top\mathbf{x}(0)$ and, as $\mathcal{A}_1\mathbf{1}=\mathbf{1}$, $\mathbf{x}(2\ell+1)=\mathcal{A}_1C^\ell\mathbf{x}(0)\rightarrow\mathcal{A}_1\mathbf{1}\pi^\top\mathbf{x}(0)=\bar{\mathbf{x}}$. If $v_r$ is a leader in the union of both layers, then $[A(2)]_r=[\mathcal{A}_1]_r=\mathbf{e}_r^\top$; hence, $[C]_r=\mathbf{e}_r^\top$, $\{v_r\}$ is a closed communication class of $C$, and $\mathcal{Z}=\{v_r\}$ by Lemma~\ref{lem:sia}.

From the standpoint of Markov chains (see \cite{bapat2010graphs,levin2017markov}), $C$ is the transition probability matrix of a chain whose transition diagram is obtained by reversing all the links of the graph of $C$; $\mathcal{Z}$ is its only absorbing class, and all other states are transient. When there is a leader, the transition probability matrix in a canonical shape, obtained by relabeling the agents so that the leader comes first and partitioning the matrix into four sub-matrices, is:
\begin{equation}\label{eqn:4}
T_1 =
\begin{pmatrix}
1 &\vline& \mathbf{0}\\
\hline
\mathbf{R} &\vline& \mathbf{Q}\\
\end{pmatrix}
\end{equation}
Where $\mathbf{R}$ corresponds to the transitions toward the absorption state, the leader's state; $\mathbf{Q}$ is the transitions of the transient states; $1$ corresponds to the leader's transition to itself, and $\mathbf{0}$ is an all-zero row vector compatible with the dimension of $\mathbf{Q}$. Obviously, $\mathbf{Q}$ is a substochastic matrix as $T_1$ is a transition probability matrix, and
\begin{align}\label{eqn:5}
T_1^\ell =
\begin{pmatrix}
1 & \vline & \mathbf{0} \\
\hline
\mathbf{R}_\ell & \vline & \mathbf{Q}^\ell \\
\end{pmatrix},\qquad \mathbf{R}_\ell=\sum_{m=0}^{\ell-1}\mathbf{Q}^m\mathbf{R}.
\end{align}
Since $T_1^\ell$ converges to the matrix whose first column is $\mathbf{1}$ and whose other entries are zero, $\lim_{\ell\rightarrow \infty}\mathbf{Q}^\ell = 0$ and $\lim_{\ell\rightarrow \infty} \mathbf{R}_\ell = \mathbf{1}$, implying that all agents will have the initial opinion of the leader.

As for the second case, instead of the leader, we have the closed communication class $\mathcal{Z}$, possibly of size greater than one, with sub-matrix $\mathbf{P}$, and the canonical form of the transition matrix is as follows:
\begin{equation}\label{eqn:6}
T_2 =
\begin{pmatrix}
\mathbf{P} &\vline& \mathbf{0}\\
\hline
\mathbf{R} &\vline& \mathbf{Q}\\
\end{pmatrix}
\end{equation}
Matrix $\mathbf{P}$ is stochastic, irreducible, and aperiodic, i.e., primitive, since $\mathcal{Z}$ is a closed, aperiodic communication class (Lemma~\ref{lem:sia}). It turns out that the limit of higher powers of this matrix is of the below form (see \cite{levin2017markov}):
\begin{align}\label{eqn:7}
    \lim_{\ell\rightarrow \infty}\mathbf{P}^\ell = \mathbf{1}\pi_{\mathcal{Z}}^\top
\end{align}
Where $\pi_{\mathcal{Z}}>0$ is the associated limiting distribution of the Markov chain with transition probability matrix $\mathbf{P}$, i.e., the restriction of $\pi$ to $\mathcal{Z}$.

Also, $\lim_{\ell\rightarrow \infty} \mathbf{Q}^\ell=\mathbf{0}_Q$, since $[\pi]_i=0$ for the transient states (Lemma~\ref{lem:sia}), where $\mathbf{0}_Q$ is a matrix of all-zero elements compatible with the size of sub-matrix $\mathbf{Q}$.
Since there is only one absorbing class, the probability of ending up in the absorbing class from any transient state will be one. Hence, the probability of ending up in a specific state of the absorbing class is its corresponding stationary distribution element, denoted as $[\pi]_i$. Thus, if matrix $\mathbf{R}$ has dimensions $m\times |\mathcal{Z}|$, where $m$ is the number of transient states, we can express the convergence as follows:
\begin{align}\label{eqn:8}
    \lim_{\ell\rightarrow \infty}T_2^\ell =
    \begin{pmatrix}
\mathbf{1}\pi_{\mathcal{Z}}^\top &\vline& \mathbf{0}\\
\hline
\mathbf{1}^{m\times 1}\pi_{\mathcal{Z}}^\top &\vline& \mathbf{0}_Q\\
\end{pmatrix}
\end{align}
This means that opinions of all the agents will finally converge to a convex combination of opinions of agents in the absorbing class of the Markov chain associated with transition probability matrix $T_2$. It is worth mentioning that these proofs also stand for the case when we have bi-directional interactions, since Assumption~\ref{as:strongly-connected} does not restrict the direction of the links.
\end{proof}

\begin{figure}[ht]
\centering
\input{ex1.tikz}
\caption{An example of a two-layer multiplex network at $t=2$ when there is a leader in the union of layers.
\label{fig:2}
}
\end{figure}

Figure~\ref{fig:2} illustrates Lemma~\ref{lem:sia} for a network with a leader. In layer 1, the nodes $A, C, E$, and $F$ act as leaders, node $E$ is isolated, and there are disconnected components; the closed communication classes $\{A\}$, $\{C\}$, $\{E\}$, and $\{F\}$ of layer 1 contain self-loops (Assumption~\ref{as:strongly-connected}.2), whereas $B$, $D$, $G$, $H$, and $I$ have no self-loop. In $\tilde{\mathcal{G}}$, the link $(D,E)$ of layer 2 is replaced by $(C,E)$, since $C$ has a self-loop and $(C,D)$ is a link of layer 1, whereas the links $(A,C)$ and $(E,F)$ of layer 2 are kept; together with the links of layer 1, $\tilde{\mathcal{G}}$ contains a directed spanning tree rooted in the leader $A$ (Assumption~\ref{as:strongly-connected}.3). The matrix $A(2)\mathcal{A}_1$ maintains a record of walks of length two, a link of layer 1 followed by a link of the union of both layers. For example, $A\rightarrow A$ in layer 1 followed by $A\rightarrow B$ in the union of both layers establishes a walk of length two from $A$ to $B$. Employing similar reasoning, the graph of $A(2)\mathcal{A}_1$ contains the links $A\rightarrow B$, $A\rightarrow C$, $C\rightarrow D$, $C\rightarrow E$, $E\rightarrow F$, $F\rightarrow G$, $F\rightarrow H$, and $G\rightarrow I$, i.e., a directed spanning tree rooted in $A$. From the standpoint of Markov chains, $A$ represents the only absorbing state, meaning it is a state in the Markov chain from which one cannot transition to another state (see \cite{levin2017markov}), and all other states are transient (to illustrate the transition diagram, one simply needs to reverse all the arrows). This implies that the final opinions of all agents converge to that of the leader.

\begin{figure}[h]
\centering
\input{ex2.tikz}
\caption{An example of a two-layer multiplex network at $t=2$ without any leader in the network.
\label{fig:3}
}
\end{figure}

In Figure~\ref{fig:3}, there is no leader in the union of both layers, and every communication class of layer 1 contains a self-loop. Hence, Assumption~\ref{as:strongly-connected}.3 reduces to a directed spanning tree in the union of both layers, e.g., $A\rightarrow B\rightarrow D\rightarrow C$ and $D\rightarrow E\rightarrow F\rightarrow G\rightarrow H\rightarrow I$. The only closed communication class of $A(2)\mathcal{A}_1$ is $\{A,B,C,D\}$, whose agents form the cycle $A\rightarrow B\rightarrow D\rightarrow C\rightarrow A$ in the union of both layers; if we flip the direction of the edges, $\{A,B,C,D\}$ is the absorbing class. Hence, the opinions converge to a convex combination of the initial opinions of $A$, $B$, $C$, and $D$.

\section{Convergence Rate}\label{sec:Convergence Rate}

In this section, we delve into the analysis of the convergence rate of the opinion dynamics model given by Equation \eqref{eqn:opinion-dynamics}, with a primary focus on insights from \cite{blondel2005convergence,olshevsky2009convergence}. For a stochastic matrix $P$, let $\lambda_2(P)$ denote the second-largest eigenvalue modulus of $P$, i.e., the largest modulus among the eigenvalues of $P$, counted with algebraic multiplicity, after one copy of the eigenvalue one is removed, and let $\rho(\cdot)$ denote the spectral radius. By Brauer's theorem \cite{horn2012matrix}, the eigenvalues of $C-\mathbf{1}\pi^\top$ are $0$ and the eigenvalues of $C$ after one copy of the eigenvalue one is removed; hence, $\lambda_2(C)=\rho(C-\mathbf{1}\pi^\top)$. Since $C\mathbf{1}\pi^\top=\mathbf{1}\pi^\top C=(\mathbf{1}\pi^\top)^2=\mathbf{1}\pi^\top$, induction gives $C^\ell-\mathbf{1}\pi^\top=(C-\mathbf{1}\pi^\top)^\ell$ for $\ell\geq 1$, which converges to zero by Lemma~\ref{lem:sia}; thus, $\lambda_2(C)<1$.

\begin{lemma}[\cite{blondel2005convergence}]\label{lem:lemma2}
Let $C$ be an SIA matrix with $\lim_{\ell\rightarrow\infty}C^\ell=\mathbf{1}\pi^\top$, and let $\mathbf{\bar{x}}=\mathbf{1}\pi^\top\mathbf{x}(0)$. For every $q\in(\lambda_2(C),1)$, there exists a constant $U>0$ such that, for all $\ell\geq 0$ and all $\mathbf{x}(0)$, the opinion dynamics characterized by the weighted adjacency matrix $C$ satisfies
$$\|C^\ell\mathbf{x}(0) - \mathbf{\bar{x}}\|_{\infty}\leq 2Uq^\ell\|\mathbf{x}(0)\|_2.$$
\end{lemma}
\begin{proof}
Since $C^\ell-\mathbf{1}\pi^\top=(C-\mathbf{1}\pi^\top)^\ell$ for $\ell\geq 1$ and $\rho(C-\mathbf{1}\pi^\top)=\lambda_2(C)<q$, Gelfand's formula \cite[Corollary 5.6.14]{horn2012matrix} implies $U':=\sup_{\ell\geq 0}\|C^\ell-\mathbf{1}\pi^\top\|_\infty q^{-\ell}<\infty$, where $\|\cdot\|_\infty$ is the maximum row sum norm. The claim follows with $U:=U'/2$, since $\|\mathbf{x}(0)\|_\infty\leq\|\mathbf{x}(0)\|_2$.
\end{proof}
By Lemma~\ref{lem:sia}, Lemma~\ref{lem:lemma2} applies to $C=A(2)\mathcal{A}_1$, although $C$ may be reducible, and $\lambda_2(C)$ can be read from the canonical forms of Section~\ref{sec:Convergence}. In the presence of a leader in the union of layers, the eigenvalues of $T_1$ in \eqref{eqn:4} are $1$ and those of $\mathbf{Q}$, so that $\lambda_2(C)=\rho(\mathbf{Q})$: the convergence rate of $A(2)\mathcal{A}_1$ is that of the corresponding component matrix $\mathbf{Q}$ towards an all-zero matrix. Without a leader, $\lambda_2(C)=\max\{\lambda_2(\mathbf{P}),\rho(\mathbf{Q})\}$ by \eqref{eqn:6}. In both cases, $\rho(\mathbf{Q})<1$ since $\lim_{\ell\rightarrow\infty}\mathbf{Q}^\ell=0$; note that $\mathbf{Q}$ need not have row sums smaller than one, e.g., for a transient agent whose neighbors are all transient.

 Therefore, the convergence rate in this case is defined as follows, drawing on \cite{blondel2005convergence, olshevsky2009convergence}, \cite[p. 341]{horn2012matrix}, with some necessary modifications:

\begin{Proposition}\label{prop:rate}
Let $q\in(\lambda_2(C),1)$ and $U$ be as in Lemma~\ref{lem:lemma2}. The convergence rate of the opinion dynamics \eqref{eqn:opinion-dynamics} is as follows:
\begin{align}
    \|\mathbf{x}(t)-\mathbf{\bar{x}}\|_{\infty} \leq 2U\|\mathcal{A}_1^{t-2\lfloor \frac{t}{2}\rfloor}\|_\infty q^{\lfloor \frac{t}{2}\rfloor}\|\mathbf{x}(0)\|_2=2Uq^{\lfloor \frac{t}{2}\rfloor}\|\mathbf{x}(0)\|_2.
\end{align}

\end{Proposition}
\begin{proof}

 \begin{align}\label{eqn:cont2}
    &\|\mathbf{x}(t)-\bar{\mathbf{x}}\|_{\infty} = \nonumber\\
    &\|A(t)\cdots A(1)\mathbf{x}(0)-\mathbf{1}\pi^\top\mathbf{x}(0)\|_{\infty}=\nonumber\\
    &\|\mathcal{A}^{t-2\lfloor\frac{t}{2}\rfloor}_1 (A(2)\mathcal{A}_1)^{\lfloor\frac{t}{2}\rfloor}\mathbf{x}(0)-\mathcal{A}^{t-2\lfloor\frac{t}{2}\rfloor}_1\mathbf{1}\pi^\top\mathbf{x}(0)\|_{\infty}\leq\nonumber\\
    &\|\mathcal{A}^{t-2\lfloor\frac{t}{2}\rfloor}_1\|_{\infty}\|(A(2)\mathcal{A}_1)^{\lfloor\frac{t}{2}\rfloor}\mathbf{x}(0)-\mathbf{1}\pi^\top\mathbf{x}(0)\|_{\infty}\leq \nonumber\\
    & 2U\|\mathcal{A}_1^{t-2\lfloor \frac{t}{2}\rfloor}\|_\infty q^{\lfloor \frac{t}{2}\rfloor}\|\mathbf{x}(0)\|_2
 \end{align}
Where we used $\mathcal{A}_1^{r}\mathbf{1}=\mathbf{1}$, $\|M\mathbf{v}\|_\infty\leq\|M\|_\infty\|\mathbf{v}\|_\infty$ for the maximum row sum norm $\|\cdot\|_\infty$, and Lemma~\ref{lem:lemma2}; $\|\cdot\|_2$ is the $\ell^2$-norm (see \cite{horn2012matrix}). Finally, $\|\mathcal{A}_1^{r}\|_\infty=1$ for every $r\geq 0$, as $\mathcal{A}_1^{r}$ is row stochastic.
\end{proof}

Since $\lfloor t/2\rfloor\geq (t-1)/2$, Proposition~\ref{prop:rate} bounds the error by a multiple of $(q^{1/2})^t$, and this per-step factor is attained as $q\downarrow\lambda_2(C)$. Indeed, as every row of $C^\ell-\mathbf{1}\pi^\top$ sums to zero, $\sup_{\mathbf{x}(0)\in[0,10]^n}\|\mathbf{x}(2\ell)-\bar{\mathbf{x}}\|_\infty=5\|C^\ell-\mathbf{1}\pi^\top\|_\infty$, whose $\ell$th root converges to $\lambda_2(C)$ by Gelfand's formula; and the error at time $2\ell+1$ is at most that at time $2\ell$. Hence, $\sigma:=\lambda_2(C)^{1/2}$ is the worst-case asymptotic convergence factor per time step.

\section{Numerical Results}\label{sec: numerical}

In this section, we (i) simulate the opinion dynamics for the examples depicted in Figure \ref{fig:2} and Figure \ref{fig:3} and their fully bidirectional counterparts, (ii) test Assumption~\ref{as:strongly-connected} on random multiplex networks, and (iii) compare the convergence speed with and without a leader, and with mono- and bidirectional interactions, on random networks with the same underlying links. The bidirectional counterpart of a network keeps its self-loops and makes every other link reciprocal, except the links leaving the leader, so that the leader remains a leader. Besides the per-step factor $\sigma=\lambda_2(C)^{1/2}$ of Section~\ref{sec:Convergence Rate}, we report the worst-case time to consensus
\begin{equation}\label{eqn:teps}
T_\varepsilon:=\min\Big\{t:\ \sup_{\mathbf{x}(0)\in[0,10]^n}\|\mathbf{x}(t)-\bar{\mathbf{x}}\|_\infty\leq 10\varepsilon\Big\},\qquad \varepsilon=10^{-3},
\end{equation}
i.e., the number of steps after which every opinion lies within $0.1\%$ of the opinion range from its limit, for every initial profile. As in Section~\ref{sec:Convergence Rate}, the supremum equals $5\|A(t)\cdots A(1)-\mathbf{1}\pi^\top\|_\infty$, which is non-increasing in $t$; hence, $T_\varepsilon$ is computed exactly and does not depend on $\mathbf{x}(0)$. We test the SIA property of $C$ exactly on the graph of $C$ (communication classes, closed classes, and greatest common divisor of cycle lengths, in integer arithmetic).

\subsection{The examples}\label{sec:num-examples}

In both examples, we consider the initial opinion profile $\mathbf{x}(0) = [4.74, 0.11, 1.14, 3.39,$ $1.16, 2.36, 0.47, 2.23, 4.92]$ of the agents $A,\ldots,I$, randomly generated from a uniform distribution. Table~\ref{tab:examples} summarizes the results, and Figure~\ref{fig:examples} shows $\|\mathbf{x}(t)-\bar{\mathbf{x}}\|_\infty$ together with the bound of Proposition~\ref{prop:rate} for $q=\lambda_2(C)+0.1(1-\lambda_2(C))$ and the constant $U=\frac{1}{2}\sup_{\ell\geq 0}\|C^\ell-\mathbf{1}\pi^\top\|_\infty q^{-\ell}$ of the proof of Lemma~\ref{lem:lemma2}.

\begin{table}[t]
\centering
\caption{The examples of Figures \ref{fig:2} and \ref{fig:3} and their bidirectional counterparts: the closed communication class $\mathcal{Z}$ of $C=A(2)\mathcal{A}_1$, the consensus value $\pi^\top\mathbf{x}(0)$, $\lambda_2(C)$, the per-step factor $\sigma$, and the worst-case time to consensus $T_\varepsilon$.}
\label{tab:examples}
\footnotesize
\begin{tabular}{llcccc}
\toprule
Network & $\mathcal{Z}$ & $\pi^\top\mathbf{x}(0)$ & $\lambda_2(C)$ & $\sigma$ & $T_\varepsilon$ \\
\midrule
Example 1, mono-directional & $\{A\}$ & 4.74 & 0.500 & 0.707 & 39 \\
Example 1, bidirectional & $\{A\}$ & 4.74 & 0.952 & 0.976 & 297 \\
Example 2, mono-directional & $\{A,B,C,D\}$ & 2.104 & 0.854 & 0.924 & 98 \\
Example 2, bidirectional & all nine agents & 2.903 & 0.922 & 0.960 & 160 \\
\bottomrule
\end{tabular}
\end{table}

\begin{figure}[t]
    \centering
    \includegraphics[width=\textwidth]{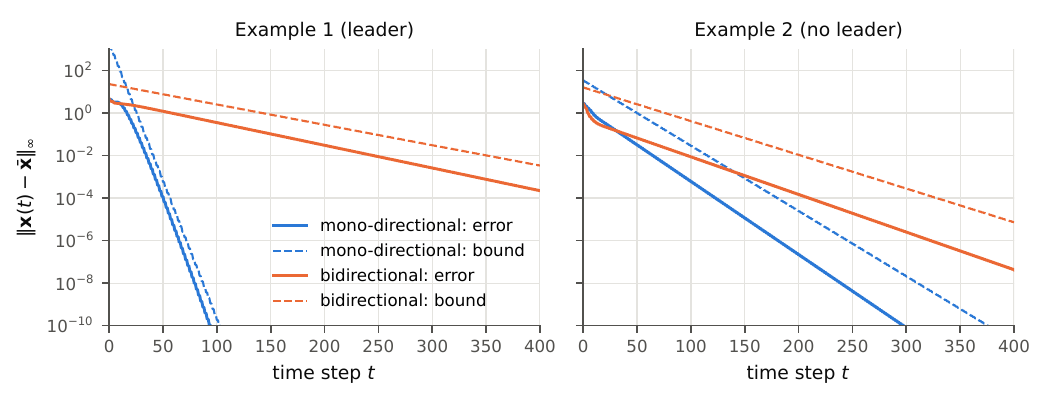}
    \caption{Distance $\|\mathbf{x}(t)-\bar{\mathbf{x}}\|_\infty$ of the opinions from their limit for the examples of Figures \ref{fig:2} (left) and \ref{fig:3} (right), with mono-directional and bidirectional interactions (solid), and the bounds of Proposition~\ref{prop:rate} (dashed).}
    \label{fig:examples}
\end{figure}

In Example 1, all agents' opinions converge to $4.74$, which is the initial opinion of the leader, aligning with Proposition~\ref{prop:convergence}, although agents $B$, $D$, $G$, $H$, and $I$ place no weight on their own opinions, i.e., the corresponding diagonal entries of $\mathcal{A}_1$ and $A(2)$ are zero. The worst-case time to consensus is $T_\varepsilon=39$ steps. In Example 2, the opinions converge to $2.104=0.2x_A(0)+0.2x_B(0)+0.4x_C(0)+0.2x_D(0)$, a convex combination of the initial opinions of the agents in the closed class $\{A,B,C,D\}$ of $C$, and consensus has been reached. The bounds of Proposition~\ref{prop:rate} are never violated and decay at nearly the same rate as the distance from the limit. In both examples, the bidirectional counterpart converges more slowly: $\sigma$ increases from 0.707 to 0.976 in Example 1 and from 0.924 to 0.960 in Example 2. In the bidirectional version of Example 2, every agent belongs to the closed class of $C$, so that all initial opinions contribute to the consensus value. Since the two examples have different topologies, they do not reveal whether the presence of a leader accelerates convergence; Section~\ref{sec:num-compare} answers this question with paired comparisons.

\subsection{Random multiplex networks and \texorpdfstring{Assumption~\ref{as:strongly-connected}}{Assumption 1}}\label{sec:num-audit}

We generate multiplex networks with $n=48$ agents from three families, which differ in layer 1:
\begin{itemize}
    \item F1 (decomposable layer 1): $n/8$ communities of seven agents, in which every ordered pair of agents is linked with probability $0.3$ and one random agent has a self-loop, and $n/8$ agents without neighbors in layer 1.
    \item F2 (no self-loops): every ordered pair of agents is linked with probability $3/(n-1)$, without self-loops.
    \item F3 (forest): $\lfloor n/5\rfloor$ agents with a self-loop are the roots of a directed forest, in which every other agent has exactly one neighbor, chosen uniformly among the agents that precede it in a random order, as in Figure~\ref{fig:2}.
\end{itemize}
In all families, every ordered pair of agents is linked in layer 2 with probability $1.5/(n-1)$, an agent without neighbors in layer 1 receives a self-loop there (Assumption~\ref{as:strongly-connected}.1), and an agent without neighbors other than itself in both layers receives one random neighbor in layer 2, so that the network has no leader. We obtain the leader variant of a network by removing all the neighbors of a random agent $v_r$ in both layers and giving it a self-loop in layer 1, and the bidirectional variants as above; hence, every instance yields four variants with the same underlying links. On 50 networks per family, the best responses to the cost \eqref{eqn:cost}, computed by minimizing $J_i$ numerically, coincide with \eqref{eqn:opinion-dynamics} up to $3\times10^{-13}$, whereas the cost without the factors $1/|\delta_{i\alpha}|$ yields best responses that deviate from \eqref{eqn:m2} by up to 2.7.

Table~\ref{tab:audit} compares Assumption~\ref{as:strongly-connected} with the SIA property of $C$ on 1,000 random networks per family, where the leader is a uniformly random agent and no network is redrawn. Assumption~\ref{as:strongly-connected} holds in 10,048 of the 12,000 variants, and $C$ is SIA in all of them, in agreement with Lemma~\ref{lem:sia}, although, e.g., in F3, only the roots of the forest and the leader have self-loops in layer 1. In contrast, a directed spanning tree in the union of both layers does not suffice: $C$ is not SIA in 612 of the 11,504 variants with such a tree, all of them with mono-directional interactions and most of them in the leader variant of F3 (405 of 800), whose layer 1 is a forest with self-loops only at the roots, as in Figure~\ref{fig:2} and in the counterexample of Section~\ref{sec:Convergence}. Assumption~\ref{as:strongly-connected} is sufficient but not necessary: $C$ is SIA in 844 variants that violate it, mostly with mono-directional interactions in F1 and F3. For an agent that lies on no cycle of layer 1, whether a link of layer 2 is used at an even time step depends on the parity of the lengths of the paths that lead to this agent, which Assumption~\ref{as:strongly-connected} does not track.

\begin{table}[t]
\centering
\caption{Assumption~\ref{as:strongly-connected} versus the SIA property of $C=A(2)\mathcal{A}_1$ on 1,000 random networks with $n=48$ agents per family, each with four variants. The second column of each group counts the variants of the first column for which $C$ is not SIA (spanning tree, Assumption~\ref{as:strongly-connected}) or Assumption~\ref{as:strongly-connected} fails ($C$ is SIA).}
\label{tab:audit}
\footnotesize
\setlength{\tabcolsep}{3.5pt}
\begin{tabular}{llcccccc}
\toprule
 & & \multicolumn{2}{c}{Spanning tree in the union} & \multicolumn{2}{c}{Assumption~\ref{as:strongly-connected}} & \multicolumn{2}{c}{$C$ is SIA} \\
\cmidrule(lr){3-4}\cmidrule(lr){5-6}\cmidrule(lr){7-8}
Family & Variant & holds & $C$ not SIA & holds & $C$ not SIA & yes & Assumption fails \\
\midrule
F1 & no leader, mono-directional & 998 & 34 & 711 & 0 & 964 & 253 \\
F1 & no leader, bidirectional & 999 & 0 & 987 & 0 & 999 & 12 \\
F1 & leader, mono-directional & 895 & 93 & 621 & 0 & 802 & 181 \\
F1 & leader, bidirectional & 942 & 0 & 930 & 0 & 942 & 12 \\
\midrule
F2 & no leader, mono-directional & 1000 & 2 & 989 & 0 & 998 & 9 \\
F2 & no leader, bidirectional & 1000 & 0 & 1000 & 0 & 1000 & 0 \\
F2 & leader, mono-directional & 988 & 5 & 978 & 0 & 983 & 5 \\
F2 & leader, bidirectional & 992 & 0 & 992 & 0 & 992 & 0 \\
\midrule
F3 & no leader, mono-directional & 999 & 73 & 793 & 0 & 926 & 133 \\
F3 & no leader, bidirectional & 1000 & 0 & 1000 & 0 & 1000 & 0 \\
F3 & leader, mono-directional & 800 & 405 & 156 & 0 & 395 & 239 \\
F3 & leader, bidirectional & 891 & 0 & 891 & 0 & 891 & 0 \\
\bottomrule
\end{tabular}
\end{table}

\subsection{Leaders and bidirectional interactions}\label{sec:num-compare}

Table~\ref{tab:compare} and Figure~\ref{fig:compare} compare the four variants on instances whose four variants satisfy Assumption~\ref{as:strongly-connected}: 200 instances per family for $n=48$, obtained from 279, 206, and 248 draws in F1, F2, and F3, where the leader is drawn uniformly among the agents for which both leader variants satisfy Assumption~\ref{as:strongly-connected}. On average, 31.0, 45.3, and 38.7 of the 48 diagonal entries of $\mathcal{A}_1$ and $A(2)$ are zero in F1, F2, and F3. In all 4,800 variants of this section (including $n=24$ and $n=96$ below), the opinions simulated from five random initial profiles in $[0,10]^n$ converge to $\pi^\top\mathbf{x}(0)$ (Proposition~\ref{prop:convergence}; maximum deviation $7.9\times10^{-11}$ once the worst-case distance is below $10^{-10}$), and the per-step factor estimated from the decay of $\sup_{\mathbf{x}(0)}\|\mathbf{x}(t)-\bar{\mathbf{x}}\|_\infty$ agrees with $\sigma=\lambda_2(C)^{1/2}$ (Section~\ref{sec:Convergence Rate}; median relative error $3.8\times10^{-6}$, 95th percentile $2.5\times10^{-2}$).

\begin{table}[t]
\centering
\caption{The four variants of 200 random instances per family with $n=48$ agents that satisfy Assumption~\ref{as:strongly-connected}: median and interquartile range of the per-step factor $\sigma=\lambda_2(C)^{1/2}$ and of the worst-case time to consensus $T_\varepsilon$, and the median size of the closed communication class $\mathcal{Z}$ of $C$.}
\label{tab:compare}
\footnotesize
\begin{tabular}{llccc}
\toprule
Family & Variant & $\sigma$ & $T_\varepsilon$ & $|\mathcal{Z}|$ \\
\midrule
F1 & no leader, mono-directional & 0.876 [0.849, 0.906] & 60 [50, 76] & 42 \\
F1 & no leader, bidirectional & 0.851 [0.833, 0.869] & 46 [42, 52] & 48 \\
F1 & leader, mono-directional & 0.992 [0.982, 0.996] & 866 [398, 1746] & 1 \\
F1 & leader, bidirectional & 0.992 [0.989, 0.995] & 879 [602, 1484] & 1 \\
\midrule
F2 & no leader, mono-directional & 0.719 [0.680, 0.750] & 24 [22, 28] & 46 \\
F2 & no leader, bidirectional & 0.567 [0.542, 0.604] & 14 [13, 16] & 48 \\
F2 & leader, mono-directional & 0.983 [0.975, 0.991] & 416 [276, 782] & 1 \\
F2 & leader, bidirectional & 0.991 [0.987, 0.994] & 747.5 [528, 1177] & 1 \\
\midrule
F3 & no leader, mono-directional & 0.862 [0.826, 0.899] & 54 [44, 74] & 14 \\
F3 & no leader, bidirectional & 0.845 [0.816, 0.879] & 46.5 [40, 58] & 48 \\
F3 & leader, mono-directional & 0.955 [0.915, 0.979] & 165.5 [94, 332] & 1 \\
F3 & leader, bidirectional & 0.989 [0.981, 0.994] & 640.5 [362, 1172] & 1 \\
\bottomrule
\end{tabular}
\end{table}

\begin{figure}[t]
    \centering
    \includegraphics[width=\textwidth]{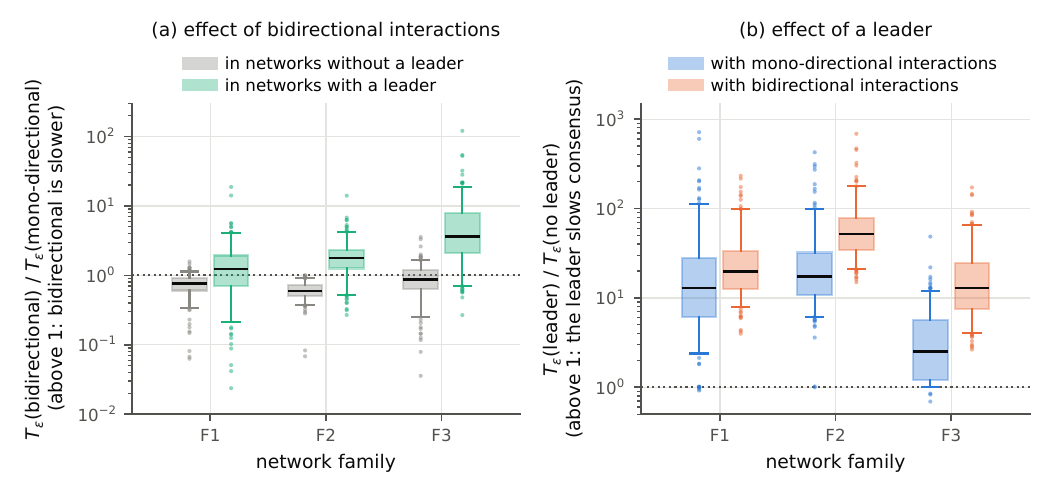}
    \caption{Paired ratios of the worst-case time to consensus $T_\varepsilon$ over 200 instances per family with $n=48$ agents: (a) effect of bidirectional interactions, i.e., $T_\varepsilon$ with bidirectional over $T_\varepsilon$ with mono-directional interactions, in networks without and with a leader; (b) effect of a leader, i.e., $T_\varepsilon$ with over $T_\varepsilon$ without a leader, with mono-directional and bidirectional interactions. Boxes show the quartiles and the median, whiskers the 5th and 95th percentiles, and dots the remaining instances; a ratio above one (dotted line) means slower convergence of the bidirectional variant in (a) and of the leader variant in (b).}
    \label{fig:compare}
\end{figure}

Leaders slow down consensus. With mono-directional interactions, the median ratio of $T_\varepsilon$ with a leader to $T_\varepsilon$ without a leader is 13.0, 17.3, and 2.5 in F1, F2, and F3, and the leader variant is slower in 196, 199, and 176 of the 200 instances, the remaining ones being mostly ties. With bidirectional interactions, the median ratios are 19.7, 51.8, and 12.8, and the leader variant is slower in all 600 instances. Without a leader, the consensus value is a weighted average over the closed class of $C$, which contains a median of 14 to 48 of the 48 agents (Table~\ref{tab:compare}); with a leader, the opinion of a single agent has to propagate to all others, and $\lambda_2(C)=\rho(\mathbf{Q})$ is close to one (Section~\ref{sec:Convergence Rate}). Hence, the faster mono-directional convergence in Example 1 than in Example 2 reflects their different topologies rather than the presence of the leader.

The effect of bidirectional interactions depends on the leader. Without a leader, reciprocal links speed up consensus: the median ratio of $T_\varepsilon$ with bidirectional interactions to $T_\varepsilon$ with mono-directional interactions is 0.77, 0.59, and 0.88 in F1, F2, and F3, and the bidirectional variant is faster in 166, 198, and 123 of the 200 instances. With a leader, reciprocal links slow consensus down: the median ratios are 1.24, 1.77, and 3.66, and the bidirectional variant is slower in 117, 165, and 180 instances, presumably because the followers then also listen to agents that are farther from the leader. All these differences are significant (two-sided sign test, $p<0.02$).

Finally, we repeated the comparison with $n=24$ and $n=96$ agents (100 instances per family each, with the same mean degrees and community sizes). Leaders slow down consensus in all families and both sizes ($p<10^{-18}$). Bidirectional interactions speed up consensus without a leader in F1 and F2 for both sizes and in F3 for $n=96$, and slow it down with a leader in F2 and F3 for both sizes ($p<10^{-8}$). The remaining effects are not significant: with a leader in F1 ($p=0.48$ for $n=24$ and $p=0.19$ for $n=96$), and without a leader in F3 for $n=24$ ($p=0.13$).

\section{Conclusions}\label{sec:Conclusions}

In conclusion, our investigation into the dynamics of opinions in a two-layer multiplex network, considering factors such as layer switching, zero diagonal elements in adjacency matrices, and decomposable structures, has provided valuable insights. We showed that a directed spanning tree in the union of both layers does not guarantee consensus when agents may place no weight on their own opinions, and we gave conditions on the layers under which the opinions converge to the initial opinion of the leader, if there is one, and otherwise to a convex combination of the initial opinions of the agents in the absorbing class.

Our study has extended the coordination game model to a multiplex network context, shedding light on the impact of active and inactive layers on the convergence of opinions. The analysis of convergence rates, aided by upper bounds derived from existing research, identifies $\lambda_2(A(2)\mathcal{A}_1)^{1/2}$ as the per-step convergence factor. Paired comparisons on random networks have revealed that a leader slows down convergence, and that bidirectional interactions introduce delays in achieving consensus in the presence of a leader but accelerate it in its absence.

Furthermore, our research contributes to the broader field of consensus problems, offering novel perspectives on networks with changing topological properties and addressing scenarios with zero diagonal elements. The upper bounds established for convergence rates provide a valuable framework for understanding the dynamics of opinion evolution in complex, multiplex networks.

Overall, our findings emphasize the intricate interplay between network structure, layer dynamics, and agent interactions, paving the way for future exploration in this evolving research domain. As we delve deeper into the complexities of multiplex networks, our study serves as a foundation for understanding the nuances of opinion dynamics and consensus formation in dynamic, real-world scenarios.

To identify potential directions for further advancements in this project, two noteworthy avenues emerge. Firstly, extending the proposed approach to networks with more than two layers could unveil additional complexities and insights into the dynamics of opinion evolution. Exploring the dynamics in other types of multilayer networks, such as interconnected networks, represents a promising frontier for comprehensive analysis.

Furthermore, incorporating stubborn agents into the network or introducing randomness in the layer-switching process could introduce realistic elements, reflecting the intricacies of real-world scenarios. These variations may offer a more nuanced understanding of opinion convergence and consensus formation in dynamic multiplex networks.

\section{Declaration of competing interest}

The authors affirm that there are no identifiable conflicting financial interests or personal associations that might have seemed to impact the work described in this article.

\section{Data and Code availability}

The code (Python with NumPy) and the data of Section~\ref{sec: numerical} accompany this manuscript; all random instances derive from a fixed seed, so that every reported number can be reproduced.

\bibliographystyle{elsarticle-num} 
\bibliography{references}
\end{document}